\documentclass{amsart}
\usepackage[utf8]{inputenc}
\usepackage{amssymb}
\usepackage{amsfonts}
\usepackage{stmaryrd}
\usepackage{amsmath}
\usepackage{xcolor}
\usepackage[english]{babel}
\usepackage{amsthm}
\usepackage{setspace}
\usepackage{tikz-cd}
\usepackage[hyphens]{url}
\usepackage{hyperref}
\usepackage[hyphenbreaks]{breakurl}
\usepackage{graphicx}
\usepackage[perpage]{footmisc}
\usepackage[shortlabels]{enumitem}
\usepackage{longtable}
\usepackage{comment}

\usepackage{mathtools}

\newtheoremstyle{example}
{}                
{}                
{\sffamily}        
{}                
{\bfseries}       
{.}               
{ }               
{}                

\usepackage[
backend=biber,
style=alphabetic,
sorting=ynt
]{biblatex}
\theoremstyle{theorem}
\newtheorem{theorem}{Theorem}

\newtheorem{lemma}[theorem]{Lemma}

\theoremstyle{example}

\theoremstyle{definition}
\newtheorem{definition}[theorem]{Definition}

\newtheorem{remark}[theorem]{Remark}

\numberwithin{theorem}{section}

\numberwithin{theorem}{section}

\DeclareMathOperator{\Hom}{Hom}

\DeclareMathOperator{\Gal}{Gal}

\DeclareMathOperator{\Aut}{Aut}

\newcommand{\R}{\mathbb{R}}
\newcommand{\Q}{\mathbb{Q}}
\newcommand{\C}{\mathbb{C}}
\newcommand{\F}{\mathbb{F}}

\usepackage{listings}
\usepackage{color}
\definecolor{keywordcolor}{rgb}{0.7, 0.1, 0.1}   
\definecolor{tacticcolor}{rgb}{0.0, 0.1, 0.6}    
\definecolor{commentcolor}{rgb}{0.4, 0.4, 0.4}   
\definecolor{symbolcolor}{rgb}{0.0, 0.1, 0.6}    
\definecolor{sortcolor}{rgb}{0.1, 0.5, 0.1}      
\definecolor{attributecolor}{rgb}{0.7, 0.1, 0.1} 

\usepackage[margin=1.2in]{geometry}
\title{Formalising the Krull Topology in Lean}
\author{Sebastian Monnet}

\begin{document}
\begin{abstract}
    The Galois group of an infinite Galois extension has a natural topology, called the \emph{Krull topology}, which has the important property of being profinite. It is difficult to talk about Galois representations, and hence the Langlands Program, without first defining the Krull topology. We explain our formalisation of this topology, and our proof that it is profinite, in the Lean 3 theorem prover. 
\end{abstract}
\maketitle

\section{Introduction}
The \emph{Langlands Program} is one of the largest and most ambitious projects in modern mathematics. The program essentially says that there is a correspondence between Galois representations and automorphic forms. Galois representations are required to be \emph{continuous}, which means that it is difficult to define them, let alone state the Langlands conjectures, without first defining the appropriate \emph{Krull topology} on Galois groups. Recent work in \cite{maria} has formalised the ad\`eles of a global field, paving the way towards the automorphic side of the Langlands philosophy. Meanwhile, we have formalised the Krull topology, laying groundwork for the Galois-theoretic side. 

One interesting feature of our formalisation is that we define the Krull topology for \emph{all} field extensions, without requiring them to be Galois. This is unusually general, and it inspired us to think about how far the abstraction could go. We found that category theory provides a natural language to express the idea in greater generality. 

The structure of the paper is as follows. In Section~\ref{section-maths}, we recap the relevant mathematics by defining field extensions, Galois groups, and the Krull topology, as well as explaining what it means for this topology to be profinite. We conclude Section~\ref{section-maths} by explaining informally a proof of profiniteness. Subsequently, in Section~\ref{section-lean}, we explain our implementation of this definition and proof, building on Lean's mathematics library, {\tt mathlib}. 

Concretely, our contributions are as follows. We created a new file, \lstinline{field_theory/krull_topology}, in {\tt mathlib}, which currently contains the definition of the Krull topology and proofs that it is Hausdorff and totally disconnected. An overview of this file's contents can be found at: \url{https://leanprover-community.github.io/mathlib_docs/field_theory/krull_topology.html}. At the time of writing, our proof that the Krull topology is profinite has not yet been merged into {\tt mathlib}, and the most recent version can be found in the Pull Request at \url{https://github.com/leanprover-community/mathlib/pull/13307}.

\section{Mathematical Preliminaries}
\label{section-maths}
We summarise the key mathematical theory underlying our work, starting with field extensions,  Galois groups, and the Krull topology, which are familiar to most number theorists. Once we have defined these objects, we move on to explaining the language of \emph{filters}, which give a convenient framework for dealing with topology. Finally, we apply these filters to prove that the Krull topology is profinite. 
\subsection{Field Theory and Galois Theory}
A \emph{field extension} $L/K$ is a pair of fields $K$ and $L$, such that $K$ is a subset of $L$. For example, the field $\Q$ of rational numbers is a subset of the real numbers $\R$, so $\R/\Q$ is a field extension. Similarly, $\C/\R$ is a field extension, where $\C$ denotes the complex numbers. For a field extension $L/K$, an element $\alpha \in L$ is said to be \emph{algebraic over $K$} if it is a root of some nonzero polynomial with coefficients in $K$. If $\alpha \in L$ is algebraic over $K$, then there is a unique monic polynomial $f_\alpha(X)\in K[X]$ of least degree such that $f_\alpha(\alpha) = 0$. This polynomial is called the \emph{minimal polynomial} of $\alpha$ over $K$. A field extension $L/K$ is said to be \emph{algebraic} if every element of $L$ is algebraic over $K$. The extension $\C/\R$ is algebraic, but $\R/\Q$ is not. 

Two important properties of some field extensions are \emph{normality} and \emph{separability}. The extension $L/K$ is said to be \emph{normal} if it is algebraic and the minimal polynomial of each $\alpha \in L$ can be factorised into linear factors over $L$. The extension $\C/\R$ is normal. For an example of an extension that is algebraic but not normal, consider the field 
$$
\Q(\sqrt[3]{2}) = \{a + b 2^{1/3} + c 2^{2/3} : a,b,c\in \Q\}
$$
as an extension of $\Q$. The minimal polynomial of $\sqrt[3]{2}$ is $X^3 - 2$, which cannot be factorised into linear factors over $\Q(\sqrt[3]{2})$, since it has roots $e^{\pm 2\pi i/3}\sqrt[3]{2}$, and these are not in $\Q(\sqrt[3]{2})$. Meanwhile, an algebraic extension $L/K$ is said to be \emph{separable} if the minimal polynomial of each $\alpha \in L$ over $K$ splits into \emph{distinct} linear factors over an algebraic closure of $L$. All algebraic extensions of $\Q$ are separable, since each minimal polynomial is coprime to its derivative, meaning that it cannot have repeated roots in any extension. 

The classic example of a non-separable algebraic extension comes from the field 
$$
L = \frac{\F_p(T)[X]}{(X^p - T)},
$$
where $p$ is prime, $\F_p$ is the field with $p$ elements, and $\F_p(T)$ is the field of rational functions over $\F_p$. Write $\sqrt[p]{T}$ for the element of $L$ represented by $X \in \F_p(T)[X]$. As the notation suggests, intuitively $L$ is obtained from $\F_p(T)$ by adjoining a $p^\mathrm{th}$ root of $T$. Then $L/\F_p(T)$ is a non-separable algebraic extension, since the minimal polynomial of $\sqrt[p]{T}$ factorises as $(y - \sqrt[p]{T})^p$ over $L$. The extension $L/K$ is said to be \emph{Galois} if it is normal and separable. 

If $L/K$ is a field extension, then $L$ naturally has the structure of a vector space over $K$. We define the \emph{degree} of this extension to be the dimension of $L$ as a $K$-vector space, and we denote it by $[L : K]$. An intuitive example is $\C/\R$. Since the complex plane is a $2$-dimensional real vector space with basis $\{1, i\}$, the degree is $[\C : \R] = 2$. A slightly more involved example is $\Q(\sqrt[3]{2})/\Q$ from before, which has degree $3$, since it has basis $\{1, 2^{1/3}, 2^{2/3}\}$ over $\Q$. 

If $L/K$ is a field extension, then an \emph{intermediate field} of $L/K$ is another field $E$ such that $K \subseteq E \subseteq L$. We will also refer to intermediate fields as \emph{subextensions} of $L/K$. In the case where a subextension $F$ of $L/K$ is of finite degree over $K$, we will call it a \emph{finite subextension}. 

\begin{definition}
    Let $L/K$ be a field extension. A \emph{$K$-algebra homomorphism} $L \to L$ is a function $\sigma : L \to L$ satisfying the following three axioms:
    \begin{enumerate}
        \item $\sigma(x + y) = \sigma(x) + \sigma(y)$ for all $x,y \in L$,
        \item $\sigma(xy) = \sigma(x)\sigma(y)$ for all $x, y\in L$,
        \item $\sigma(x) = x$ for all $x \in K$.
    \end{enumerate}
    If moreover $\sigma:L\to L$ is a bijection, then it is called a \emph{$K$-algebra isomorphism}.
\end{definition}
\begin{definition}
    Let $L/K$ be any field extension. We define the \emph{Galois group} $\Gal(L/K)$ of $L/K$ to be the set of $K$-algebra isomorphisms $\sigma:L \to L$, which is a group under composition. 
\end{definition}
\begin{remark}
    It is slightly unconventional to define the Galois group of a field extension that is not a Galois extension. Usually, this object would be denoted $\Aut_K(L)$. However, in {\tt{mathlib}}, both objects are represented by the same notation, and all of our Lean definitions and results apply to non-Galois extensions. Since it will not matter to us whether an extension is Galois, we will use the notation $\Gal(L/K)$ for all extensions $L/K$. 
\end{remark}
In general, one may define a pair of maps 
$$
\begin{tikzcd}[column sep=large]
    \{\text{subgroups of $\Gal(L/K)$}\} \arrow[r, "H \mapsto L^H", shift left] & \{\text{intermediate fields of $L/K$}\}, \arrow[l, "\Gal(L/E) \mapsfrom E", shift left]
    \end{tikzcd}
$$
where
\begin{align*}
L^H &:= \{x \in L : \sigma(x) = x \text{ for all $\sigma \in H$}\}.
\end{align*}
Note that for an intermediate field $E$ of $L/K$, the group $\Gal(L/E)$ is indeed a subgroup of $\Gal(L/K)$, since an isomorphism of $L$ fixing $E$ certainly also fixes $K$. 

We call $L^H$ the \emph{fixed subfield} of $H$, since it consists of the elements of $L$ that are fixed by $H$. Similarly, when viewed as a subgroup of $\Gal(L/K)$, the group $\Gal(L/E)$ is called the \emph{fixing subgroup} of $E$, since it consists of the elements of $\Gal(L/K)$ fixing $E$.

One reason to care about Galois groups is the following theorem, which is a special case of \cite{MR2180311}, Page 120, Theorem 7.34.
\begin{theorem}[Fundamental Theorem of Galois Theory]
    If $L/K$ is a Galois extension of finite degree, then the maps $H\mapsto L^H$ and $E\mapsto \Gal(L/E)$ are mutually inverse bijections. 
\end{theorem} 
 
\subsection{The Krull Topology}

The Fundamental Theorem of Galois Theory breaks down for infinite Galois extensions. See \cite{conrad}, Examples 3.10 and 3.11 for counterexamples. To salvage the theorem, we define a topology on $\Gal(L/K)$.

Let $L/K$ be a Galois extension, possibly of infinite degree. Recall that a \emph{finite subextension of $L/K$} is an intermediate field $F$ such that $F/K$ is of finite degree. 
\begin{definition} 
    We define the \emph{Krull topology} on $\Gal(L/K)$ to be the topology generated by sets of the form 
    $$
    \sigma \Gal(L/F) := \{\sigma f : f \in \Gal(L/F)\},
    $$ 
    where $\sigma \in \Gal(L/K)$ and $F/K$ is a finite subextension of $L/K$. 
\end{definition} 
If cosets $\sigma \Gal(L/F_1)$ and $\tau \Gal(L/F_2)$ have nonempty intersection, then for every $\varphi \in \sigma \Gal(L/F_1) \cap \tau \Gal(L/F_2)$, we have 
$$
\varphi \in \varphi \Gal(L/F_1F_2) \subseteq \sigma \Gal(L/F_1) \cap \tau \Gal(L/F_2),
$$
where $F_1F_2$ is the smallest subfield of $L$ containing $F_1$ and $F_2$, which is also of finite degree over $K$. Therefore, the open sets of the Krull topology are precisely the unions of cosets of the form $\sigma \Gal(L/F)$ for finite subextensions $F$ of $L/K$. 
\begin{definition}
    A \emph{topological group} is a group $G$, equipped with a topology, such that the maps 
    \begin{align*}
        &G\times G \to G,\quad G \to G
        \\
        &(x,y) \mapsto xy, \quad\text{ } x \mapsto x^{-1}
    \end{align*}
    are both continuous. 
\end{definition}
In the remainder of Section~\ref{section-maths}, we will state several results whose proofs are elementary. Since the focus of the paper is on computation, we will omit these elementary proofs in order to spend more time discussing our implementation. 
\begin{lemma}
    Let $L/K$ be a (possibly infinite) Galois extension. The Krull topology makes $\Gal(L/K)$ into a topological group. 
\end{lemma}

The following theorem salvages the Fundamental Theorem of Galois Theory for infinite extensions. 
\begin{theorem}[Krull]
    The mappings $E \mapsto \Gal(L/E)$ and $H \mapsto L^H$ are mutually inverse bijections between intermediate fields of $L/K$ and closed subgroups of $\Gal(L/K)$. 
\end{theorem}
\begin{proof}
    This is Part (1) of \cite{berhuy_2010}, Theorem I.2.8. 
\end{proof}
\begin{definition}
    A \emph{profinite group} is a topological group that is isomorphic to the limit of an inverse system of finite groups, with their discrete topologies. 
\end{definition}
As $F$ ranges over finite normal subextensions $F/K$, the groups $\Gal(F/K)$ form an inverse system of finite topological groups, equipped with their discrete topologies. The restriction maps $\Gal(L/K)\to \Gal(F/K)$ make the group $\Gal(L/K)$ into the limit of this inverse system, in the category of topological groups. 

Therefore, $\Gal(L/K)$ is a profinite group. The following theorem gives a convenient explicit condition for profiniteness. 
\begin{theorem}
    \label{theorem-profinite-equiv-characterisations}
    A topological group is profinite if and only if its topology is compact, Hausdorff, and totally disconnected. 
\end{theorem}
\begin{proof}
    This is Theorem 2 of \cite{shatz}.
\end{proof}
In {\tt mathlib}, a space is \emph{defined} to be profinite if it is compact, Hausdorff, and totally disconnected, so those are the conditions we proved.

\subsection{Filters and Filter Bases}
Our proof of profiniteness uses the language of filters, which we now explain. 
\begin{definition}
    A \emph{filter} on a set $X$ is a collection $\mathcal{F}$ of subsets of $X$, satisfying the following axioms: 
    \begin{enumerate}
        \item (Universality) $X \in \mathcal{F}$,
        \item (Upward closure) if $S \in \mathcal{F}$ and $S \subseteq T \subseteq X$, then $T \in \mathcal{F}$,
        \item (Closure under intersection) if $S, T \in \mathcal{F}$, then $S\cap T \in \mathcal{F}$. 
    \end{enumerate}
\end{definition}
\begin{definition}
    A \emph{filter bundle} on a set $X$ is a function 
    $$
    \mathcal{N} : X \to \{\text{filters on $X$}\}.
    $$
\end{definition}

\begin{lemma}
    \label{lem-top-induced-by-filters}
    Let $\mathcal{N}$ be a filter bundle on a set $X$. Define 
    $$
    \mathcal{T} = \{U \subseteq X : U \in \mathcal{N}(x) \text{ for all $x \in U$}\}.
    $$
    Then $\mathcal{T}$ is a topology on $X$. 
\end{lemma}
The topology from Lemma~\ref{lem-top-induced-by-filters} is called the topology \emph{induced by $\mathcal{N}$}. 
We also introduce the concept of a filter basis, which is a reduced set of data that generates a filter. 
\begin{definition}
    A \emph{filter basis} on a set $X$ is a collection $\mathcal{B}$ of subsets of $X$, satisfying the following two axioms:
    \begin{enumerate}
        \item $\mathcal{B} \neq \varnothing$,
        \item for all $U,V\in\mathcal{B}$, there is some $W \in \mathcal{B}$ with $W \subseteq U \cap V$. 
    \end{enumerate}
\end{definition}
\begin{lemma}
    \label{lem-induced-filter-from-basis}
Let $\mathcal{B}$ be a filter basis on a set $X$. The collection 
$$
\mathcal{F} = \{U \subseteq X : D \subseteq U \text{ for some } D \in \mathcal{B}\} 
$$
is a filter on $X$. 
\end{lemma}
The filter $\mathcal{F}$ from Lemma~\ref{lem-induced-filter-from-basis} is called the \emph{filter induced by $\mathcal{B}$}. 
\subsection{Group Filter Bases}
At this point, we can come up with a strategy for obtaining a topology on a group from a filter basis. The na\"ive strategy is as follows:
\begin{enumerate}
    \item Start with a filter basis $\mathcal{B}$ on a group $G$.
    \item For each $g \in G$, write $g\cdot \mathcal{B}$ for the collection of cosets $\{g\cdot D : D \in \mathcal{B}\}$. 
    \item It is easy to see that each $g\cdot\mathcal{B}$ is a filter basis on $G$, so we may define $\mathcal{N}(g)$ to be the filter induced by $g\cdot \mathcal{B}$. 
    \item Then $\mathcal{N}$ is a well-defined filter bundle on $G$, so it induces a topology. 
\end{enumerate}

Let $\mathcal{B}$ be a filter basis on a group $G$, and define a topology on $G$ as outlined above. This topology is called the \emph{topology induced by $\mathcal{B}$}. 

Each filter basis on $G$ does induce a topology, but that topology will not in general make $G$ into a topological group, since multiplication and inversion may not be continuous. We can remedy this by imposing the following additional conditions on our filter basis.
\begin{definition}
    \label{def-group-filter-basis}
    Let $G$ be a group. A \emph{group filter basis} on $G$ is a filter basis $\mathcal{B}$ on $G$ such that:
    \begin{enumerate}
        \item $1 \in U$ for all $U \in \mathcal{B}$,
        \item for all $U \in \mathcal{B}$, there is some $V \in \mathcal{B}$ with $V\cdot V \subseteq U$. 
        \item for all $U \in \mathcal{B}$, there is some $V \in \mathcal{B}$ with $V \subseteq U^{-1}$,
        \item for all $x \in G$ and $U \in \mathcal{B}$, there is some $V \in \mathcal{B}$ with $x V x^{-1} \subseteq U$. 
    \end{enumerate}
\end{definition}
Group filter bases were first defined in \cite{MR1726779}, Pages 222-223, and the discussion preceding that definition explains how to define the induced group topology on $G$, as we have done here. In {\tt mathlib}, group filter bases were formalised by Patrick Massot. 

It turns out that these are the conditions we need for the induced topology on $G$ to make $G$ into a topological group, as we see from the following theorem. The proof is trickier than most of the surrounding lemmas, but still elementary, so we omit it. 

\begin{theorem}
    Let $\mathcal{B}$ be a group filter basis on a group $G$. Then the topology induced by $\mathcal{B}$ makes $G$ into a topological group. 
\end{theorem}

\begin{lemma}
    \label{lem-gfb}
    Let $L/K$ be any extension of fields, and define the set 
    $$
    \mathcal{B} = \{\Gal(L/F) : F/K \text{ is a finite subextension of $L/K$}\}.
    $$
    Then $\mathcal{B}$ is a group filter basis for $\Gal(L/K)$.
\end{lemma}
For an extension $L/K$, the group filter basis from Lemma~\ref{lem-gfb} is called the \emph{standard group filter basis} on $\Gal(L/K)$. 
\begin{lemma}
    Let $L/K$ be a Galois extension of fields. The topology induced by the standard group filter basis on $\Gal(L/K)$ is equal to the Krull topology. 
\end{lemma}
\subsection{Proof that the Krull Topology is Profinite}
\label{subsection-proof-of-profinite-informal}
We have written Lean proofs that the Krull topology is Hausdorff, totally disconnected, and compact. The former two properties are elementary, so we will not explain their proofs here. Proving compactness is more difficult, and this subsection is devoted to explaining its proof informally. 

Compactness has an equivalent characterisation involving filters. Before we can state this criterion (which we do in Theorem~\ref{theorem-compact-iff-ultrafilter}), we need two more definitions. 
\begin{definition}
    Let $X$ be a topological space. An \emph{ultrafilter on $X$} is a filter $\mathcal{F}$ on $X$ satisfying the following two axioms:
    \begin{enumerate}
        \item $\varnothing \not \in \mathcal{F}$,
        \item for any filter $\mathcal{G}$ on $X$ with $\varnothing\not\in\mathcal{G}$ and $\mathcal{F}\subseteq \mathcal{G}$, we have $\mathcal{G} = \mathcal{F}$.
    \end{enumerate}
\end{definition}
\begin{definition}
    Given a point $x$ of a topological space $X$, the \emph{neighbourhood filter of $x$} is the filter 
    $$
    \mathcal{N}(x) = \{N \subseteq X: \text{ there is some open $U\subseteq X$ with $x \in U \subseteq N$}\}.
    $$
\end{definition}
The following theorem gives a convenient equivalent condition for a topological space to be compact. 
\begin{theorem}
    \label{theorem-compact-iff-ultrafilter}
    Let $X$ be a topological space. The following are equivalent:
    \begin{enumerate}
        \item $X$ is compact, 
        \item for every ultrafilter $\mathcal{F}$ on $X$, there is some $x \in X$ such that $\mathcal{N}(x) \subseteq \mathcal{F}$. 
    \end{enumerate}
\end{theorem}
\begin{proof}
    This is Theorem 2.2.5 of \cite{MR1932358}.
\end{proof}
Following Theorem~\ref{theorem-compact-iff-ultrafilter}, the key idea in our proof of compactness is as follows: take an arbitrary ultrafilter $\mathcal{F}$ on $\Gal(L/K)$ and construct an element $\sigma \in \Gal(L/K)$ such that $\mathcal{N}(\sigma) \subseteq \mathcal{F}$. For each point $x\in L$ and each finite normal subextension $F/K$ containing $x$, we will construct a $K$-algebra homomorphism $\varphi_{F,x}:F \to L$, satisfying certain properties. We will then glue these ``local'' $K$-algebra homomorphisms to define a ``global'' $K$-algebra isomorphism $\sigma:L\to L$. 

So, let $\mathcal{F}$ be an ultrafilter on $\Gal(L/K)$. Let $x \in L$ and let $F/K$ be a finite subextension containing $x$. There is a restriction map $p :\Gal(L/K) \to \Hom_K(F,L)$, and we obtain a \emph{pushforward} $p_*\mathcal{F}$ of $\mathcal{F}$ along $p$, which consists of sets whose preimages under $p$ are in $\mathcal{F}$. It is easy to see that the pushforward of an ultrafilter is an ultrafilter, and also that ultrafilters on finite sets are always principal\footnote{We have not officially defined this. It just means that the ultrafilter consists of all sets containing a designated element, called the ``generator'' of the ultrafilter.}. It follows that $p_*\mathcal{F}$ is generated by some element ${\varphi}_{F,x} \in \Hom_K(F,L)$. It turns out that the element $\varphi_{F,x}(x)\in L$ is independent of the choice of $F$, so there is a well-defined element $\sigma(x) \in L$ such that $\sigma(x) = \varphi_{F,x}(x)$ for all finite subextensions $F$ containing $x$. This allows us to define a function $\sigma:L \to L$. 

Now, for any elements $x,y \in L$, setting $F = K(x,y)$ shows that $\sigma(x + y) = \sigma(x) + \sigma(y)$, and similarly for products. Therefore, $\sigma$ is a ring homomorphism. Clearly $\sigma$ fixes $K$, so it is a $K$-algebra homomorphism. 

We need to show that $\sigma$ is actually a $K$-algebra isomorphism $L \to L$. Field homomorphisms are always injective, so we only need to show that $\sigma$ is surjective. Let $y \in L$ and write $f_y$ for its minimal polynomial over $K$. Let $y_1, y_2,\ldots, y_n$ be the roots of $f_y$ in $L$, and let $F = K(y_1,\ldots, y_n)$. Then $F/K$ is a finite extension and $\varphi_{F,y}$ is a $K$-algebra homomorphism $F \to F$, which means that $\varphi_{F,y}$ gives an isomorphism $F \to F$. In particular, there is some $x \in F \subseteq L$ such that $\varphi_{F,y}(x) = y$, so $\sigma(x) = y$. Therefore, $\sigma$ is surjective. 

Now, if $U \in \mathcal{N}(\sigma)$, then by definition of the Krull topology, there is some finite subextension $F/K$ with $\sigma\cdot \Gal(L/F) \subseteq U$. Now, if $p:\Gal(L/K)\to \Hom_K(F,L)$ is the restriction map, then by definition of $\sigma$, the ultrafilter $p_*\mathcal{F}$ is generated by $\sigma|_F$, which means that $\{\sigma|_F\} \in p_*\mathcal{F}$. By definition of pushforwards, the set 
$$
p^{-1}(\{\sigma|_F\}) = \sigma\Gal(L/F)
$$
is in $\mathcal{F}$, which means that $U$ is too, by upward closure. Therefore we have $\mathcal{N}(\sigma) \subseteq \mathcal{F}$, so Theorem~\ref{theorem-compact-iff-ultrafilter} tells us that $\Gal(L/K)$ is a compact topological space. 

\section{Implementation in Lean}
\label{section-lean}

In this section, we give an overview of our implementation in Lean. We start by explaining how we defined the Krull topology and then discuss our proof that it is profinite. 
\subsection{Type Theory Basics}
Instead of set theory, Lean is based on the formalism of \emph{dependent type theory}. Roughly speaking, one can think of a \emph{type} as a collection of things, like a set. Instead of being called elements, the things inside a type are called \emph{terms}. Given a type \lstinline{X}, we write \lstinline{x : X} to say that \lstinline{x} is a term of type \lstinline{X}. In keeping with the analogy to sets, we essentially write \lstinline{:} to mean $\in$. 

For any pair \lstinline{X, Y} of types, there is another type \lstinline{X → Y} of \emph{functions from \lstinline{X} to \lstinline{Y}}. These do what we expect them to; they assign a term of \lstinline{Y} to each term of \lstinline{X}. 

There is a special sort of type called a \emph{proposition}. A proposition can have at most one term. If the proposition has a term, we say that it is \emph{true}, and \emph{false} otherwise. If a proposition \lstinline{P} is true, then we call the unique term \lstinline{p : P} the \emph{proof of \lstinline{P}}. The fact that propositions have at most one term is called \emph{proof irrelevance}, and it is a foundational design choice of Lean. 

Moreover, for any propositions \lstinline{P} and \lstinline{Q}, the type \lstinline{P → Q} of functions from \lstinline{P} to \lstinline{Q} is also a proposition. If the proposition \lstinline{P → Q} is true (i.e. if there exists a function from \lstinline{P} to \lstinline{Q}), then we say that \lstinline{P} \emph{implies} \lstinline{Q}.  Note that this is just a formalism - it is not obvious that these notions of propositions, truth, and implication actually align with our conventional understanding of the words. It turns out though that traditional logic is naturally emergent from the formalism. 

For example, suppose that we have propositions \lstinline{P,Q}, and \lstinline{R}, such that \lstinline{P} implies \lstinline{Q} and \lstinline{Q} implies \lstinline{R}. By definition, the propositions \lstinline{P → Q} and \lstinline{Q → R} are true. Let \lstinline{f : P → Q} and \lstinline{g : Q → R} be the proofs of these propositions. Then, composing functions, we obtain a term \lstinline{g ∘ f} of type \lstinline{P → R}. That is, we have constructed a proof of \lstinline{P → R}, which means that \lstinline{P → R} is true, so \lstinline{P} implies \lstinline{R}. Therefore, we have recovered the intuitive notion of transitivity of implication from our formal definitions. 

\subsection{Definition of the Krull Topology}
Let $L/K$ be a field extension, not necessarily Galois. In Lean, we define the Krull topology on $\Gal(L/K)$ to be the topology generated by the standard group filter basis on $\Gal(L/K)$. In the {\tt mathlib} API, the Galois group of $L/K$ is denoted \lstinline{L ≃ₐ[K] L}.

\begin{remark} 
    \label{rmk-function-vs-alg-hom}
    Understanding the notation \lstinline{L ≃ₐ[K] L} sheds light on the way that {\tt{mathlib}} is organised, so we will take a  moment to explain it. In {\tt{mathlib}}, we write \lstinline{L → L} for the type of functions from $L$ to $L$. These functions do not see the additional structure of $L$ as a $K$-algebra. On the other hand, we write \lstinline{L →ₐ[K] L} for the type of $K$-algebra homomorphisms from $L$ to $L$. Mathematicians typically think of an algebra homomorphism as being ``the same thing'' as its underlying function; a homomorphism is just a function satisfying some additional conditions. In {\tt{mathlib}}, however, an algebra homomorphism is a different object from its underlying function. The algebra homomorphism \emph{contains} the underlying function, along with proofs that the function satisfies the required properties. 

    Finally then, \lstinline{L ≃ₐ[K] L} denotes the type of $K$-algebra \emph{equivalences} from $L$ to $L$. Again, in {\tt{mathlib}}, a $K$-algebra equivalence is different from a bijective homomorphism. A term of \lstinline{L ≃ₐ[K] L} consists of the following data:
    \begin{enumerate}
        \item A term \lstinline{to_fun} of type \lstinline{L → L},
        \item a term \lstinline{inv_fun} of type \lstinline{L → L},
        \item proofs that \lstinline{to_fun} and \lstinline{inv_fun} are mutual inverses, and also that they satisfy the properties of $K$-algebra homomorphisms.
    \end{enumerate} 
    These different data structures give insight into why formalising mathematics is difficult. When using Lean, we often have to keep track of distinctions between objects that we intuitively consider to be ``the same'', but whose implementations are fundamentally different objects. 
\end{remark}

For each intermediate field $E$ of $L/K$, the subgroup of terms $\sigma$ of \lstinline{L ≃ₐ[K] L} fixing $E$ is called \lstinline{E.fixing_subgroup}. Upon encountering this notation for the first time, one might ask how Lean knows that \lstinline{E.fixing_subgroup} is in fact a subgroup of \lstinline{L ≃ₐ[K] L}. That is, the subgroup depends on $L$ and $K$, whereas its definition mentions only $E$. The explanation is that $E$ is a term of type \lstinline{intermediate_field K L}, so we can recover $K$ and $L$ by looking at the type of $E$. 
\begin{remark} 
    Note that \lstinline{L ≃ₐ[K] L} and \lstinline{K.fixing_subgroup} are different objects in {\tt mathlib}; although they both represent the Galois group of $L/K$, the former has type \lstinline{Type u} for some universe \lstinline{u}, while the latter has type \lstinline{set(L ≃ₐ[K] L)}.
\end{remark}
In order to define the standard group filter basis, we define \lstinline{finite_exts K L} to be the set of intermediate fields $F$ of $L/K$ such that $F/K$ is finite dimensional: 
\begin{lstlisting}
def finite_exts (K : Type*) [field K] (L : Type*) [field L] 
  [algebra K L] :
  set (intermediate_field K L) :=
{E | finite_dimensional K E}
\end{lstlisting}
Subsequently, we define the set \lstinline{fixed_by_finite K L} consist of subsets of the form \lstinline{F.fixing_subgroup}, as \lstinline{F} ranges over \lstinline{finite_exts K L}:
\begin{lstlisting}
def fixed_by_finite (K L : Type*) [field K] [field L] 
  [algebra K L]: set (subgroup (L ≃ₐ[K] L)) :=
intermediate_field.fixing_subgroup '' (finite_exts K L)
\end{lstlisting}
\begin{remark}
    The \lstinline{''} in the definition of \lstinline{fixed_by_finite} denotes the image of a set under a function. In general, if \lstinline{X} and \lstinline{Y} are types, \lstinline{f : X → Y} is a function, and \lstinline{S : set X} is a set of terms of \lstinline{X}, then \lstinline{f '' S} denotes the image of \lstinline{S} under the function \lstinline{f}. 
\end{remark}

The elements of \lstinline{fixed_by_finite K L} are then precisely the elements of the standard group filter basis. However, as far as Lean is concerned, \lstinline{fixed_by_finite K L} \emph{is not} a group filter basis, but merely a set of subgroups of \lstinline{L ≃ₐ[K] L}. In {\tt{mathlib}}, a term of type \lstinline{group_filter_basis K L} consists of the following data:
\begin{enumerate}
    \item A term of type \lstinline{filter_basis K L},
    \item four proofs, showing that the filter basis in question satisfies the additional axioms of a group filter basis.
\end{enumerate}
In turn, a term of type \lstinline{filter_basis K L} consists of:
\begin{enumerate}
    \item A term of type \lstinline{set (L ≃ₐ[K] L)},
    \item two proofs, showing that the set in question satisfies the axioms of a filter basis. 
\end{enumerate}
The underlying set of our filter basis is \lstinline{fixed_by_finite K L}, which we package into the filter basis \lstinline{gal_basis K L}. Subsequently, we use \lstinline{gal_basis K L} as the underlying filter basis of the group filter basis \lstinline{gal_group_basis K L}. This hierarchy is somewhat awkward to write out in prose, so we summarise it in the following diagram:
$$
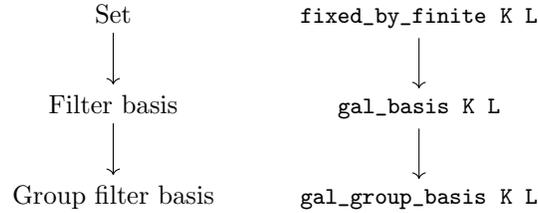
\begin{tikzcd}
    \text{Set} \arrow[d]          & \arrow[d] \text{\lstinline{fixed_by_finite K L}} \\
    \text{Filter basis} \arrow[d] & \arrow[d] \text{\lstinline{gal_basis K L}}      \\
    \text{Group filter basis}     & \text{\lstinline{gal_group_basis K L}}
    \end{tikzcd}
$$
\begin{remark}
    \label{rmk-set-vs-filter-basis-vs-grp-filter-basis}
It is important to keep track of the different types of these terms, since Lean considers them to be different objects. This can be unintuitive for mathematicians, since we would generally consider the (group) filter basis to be \emph{the same thing} as its underlying set. 
\end{remark}

We can now define the Krull topology, \lstinline{krull_topology K L}, on the group \newline \lstinline{L ≃ₐ[K] L} by:
\begin{lstlisting}
instance krull_topology (K L : Type*) [field K] [field L] 
  [algebra K L] :
  topological_space (L ≃ₐ[K] L) :=
group_filter_basis.topology (gal_group_basis K L)
\end{lstlisting}
\begin{remark}
    We defined \lstinline{krull_topology K L} as an instance, which means that the type class inference system understands it as ``the'' topology on a Galois group. This means that we can make topological statements about subsets of \lstinline{L ≃ₐ[K] L}, and the elaborator will automatically infer that we are talking about the Krull topology. 
\end{remark}
\begin{remark}
    \label{rmk-category-theory}
    Our definition of the Krull topology is valid for any field extension, not necessarily normal or separable. This is more general than definitions in the literature, which made us think about how far the generalisation could go. In fact, it can be taken much further, as we now explain. Fix a category $\mathcal{C}$, and let $L$ be an object of $\mathcal{C}$. There is a category $\operatorname{subobj}(L)$ of \emph{subobjects} of $L$. The objects of $\operatorname{subobj}(L)$ are pairs $(E,i)$, where $E \in \mathcal{C}$ and $i:E \hookrightarrow L$ is a monomorphism. The morphisms of $\operatorname{subobj}(L)$ are the obvious commutative triangles. Suppose further that there is a set $S$ of objects of $\operatorname{subobj}(\mathcal{C})$, whose members we will call \emph{small} objects of $\mathcal{C}$, with the following two axioms:
    \begin{enumerate}
        \item (The intersection of subobjects contains a subobject). For all 
        $$
        (E_1, i_1), (E_2, i_2) \in S,
        $$ 
        there is some $(E,i)\in S$ and maps $f_j : E \to E_j$ such that the diagram 
        $$
        \begin{tikzcd}
            & E_1 \arrow[d, "i_1"]  \\
E \arrow[ru, "f_1", dashed] \arrow[rd, "f_2"', dashed] \arrow[r, "i"] & L                     \\
            & E_2 \arrow[u, "i_2"']
\end{tikzcd}
        $$
        commutes.
        \item (Automorphisms preserve subobjects). For all $(E,i) \in S$ and all $\sigma \in \Aut_\mathcal{C}(L)$, we also have $(E, \sigma\circ i) \in S$.
    \end{enumerate}
    For each $(E,i) \in \operatorname{subobj}(\mathcal{C})$, define the \emph{fixing subgroup} of $(E,i)$ to be $F(E,i)\subseteq \Aut_\mathcal{C}(L)$ to be 
    $$
    F(E,i) = \{\sigma \in \Aut_\mathcal{C}(L) : \sigma\circ i = i\}. 
    $$
    Then the collection 
    $$
    \mathcal{B} = \{F(E,i) : E \in S\}
    $$
    is a group filter basis for $\Aut_\mathcal{C}(L)$, so it gives it the structure of a topological group. The Krull topology is a special case of this construction, where $\mathcal{C}$ is the category of $K$-algebras and $L$ is viewed as an object of $\mathcal{C}$. In that case, Axiom (1) comes from the fact that the intersection of two $K$-subalgebras of $L$ is a $K$-subalgebra of $L$, and Axiom (2) comes from the fact that $\sigma(E)$ is a $K$-subalgebra of $L$, for all $K$-subalgebras $E$ and all $K$-isomorphisms $\sigma:L\to L$. We did not treat this level of generality in Lean, since it seemed quite far-removed from our objective. 
\end{remark}
A surprising difficulty was in proving that \lstinline{fixed_by_finite K L} satisfies Axiom $(2)$ of Definition~\ref{def-group-filter-basis}, which is equivalent to proving that the join of two finite-dimensional field extensions is finite-dimensional. Mathematically, this is very simple; if $F/K$ and $E/K$ are finite extensions with bases $\{x_i\}$ and $\{y_j\}$ respectively, then the products $\{x_iy_j\}$ form a finite spanning set for the join $FE$, so $FE$ is finite-dimensional over $K$. Since this result is so elementary, we assumed that it must already be in {\tt mathlib}, so we asked in the Zulip chat\footnote{Much of Lean's community uses a dedicated server on the chat and collaborative software, Zulip.} about where it might be. It turned out that the result was not in {\tt mathlib}, and that proving it in Lean was actually quite difficult. Thomas Browning generously wrote the proof, and it is now in {\tt mathlib} under the name \lstinline{intermediate_field.finite_dimensional_sup}. 

The difficulty stemmed from the fact that finiteness is hard to formalise, and there are numerous ways to approach it in Lean. For example, given a type \lstinline{X}, there are types \lstinline{list X}, \lstinline{multiset X}, and \lstinline{finset X}, which all capture slightly different notions of ``a finite collection of terms of \lstinline{X}'', depending on whether we care about ordering and multiplicity. There is also a type called \lstinline{fintype X}, which has a term if and only if \lstinline{X} contains finitely many elements. On the other hand, given a term \lstinline{s} of type \lstinline{set X}, there is a proposition \lstinline{finite s}, which says that the set is finite. Each approach has pros and cons, and choosing the right tool for the job can be difficult. Moreover, we often have to manage interactions between multiple notions of finiteness, which requires a clear understanding of the relationships between them. 

\subsection{Proof of Profiniteness}
The proofs that $\Gal(L/K)$ is Hausdorff and totally disconnected are straightforward, so we will not say much about them. They are formalised in 
\begin{lstlisting} 
lemma krull_topology_t2 {K L : Type*} [field K] [field L] 
    [algebra K L] (h_int : algebra.is_integral K L) : 
t2_space (L ≃ₐ[K] L) :=
\end{lstlisting}
and 
\begin{lstlisting}
lemma krull_topology_totally_disconnected {K L : Type*} [field K] 
    [field L] [algebra K L] (h_int : algebra.is_integral K L) : 
is_totally_disconnected (set.univ : set (L ≃ₐ[K] L)) :=
\end{lstlisting}
\begin{remark}
    In the lemmas above, \lstinline{h_int} is a term of type \lstinline{is_integral K L}. Therefore, the Krull topology is Hausdorff and totally disconnected whenever the extension $L/K$ is algebraic. Note that normality and separability do not enter the picture. 
\end{remark}

Our proof of compactness, which we explained informally in Section~\ref{subsection-proof-of-profinite-informal}, is more involved. Given an ultrafilter $\mathcal{F}$ on $\Gal(L/K)$, a finite normal subextension $F/K$, and an element $x \in F$, we defined a $K$-algebra homomorphism ${\varphi}_{F,x}:F \to L$. We then glued the various $\varphi_{F,x}$ together to obtain a map $\sigma \in \Gal(L/K)$ with $\mathcal{N}(\sigma) \subseteq \mathcal{F}$. In Lean, the homomorphism ${\varphi}_{F,x}$ is defined by 
\begin{lstlisting}
protected noncomputable def ultrafilter.generator_of_pushforward
    (h_findim : finite_dimensional K E) (f : ultrafilter (L →ₐ[K] L)) : 
E →ₐ[K] L :=
classical.some $ ultrafilter.eq_principal_of_fintype _ $
    f.map $ λ σ, σ.comp $ intermediate_field.val _
\end{lstlisting}
\begin{remark}
    The definition above is labelled as \lstinline{noncomputable} because it uses Lean's axiom of choice. In particular, 
    \begin{lstlisting}
    $ ultrafilter.eq_principal_of_fintype _ $
        f.map $ λ σ, σ.comp $ intermediate_field.val _
    \end{lstlisting}
    \vspace{-1em}
    is a term of type 
    \begin{lstlisting}
    ∃ (x : F →ₐ[K] L), ↑(ultrafilter.map (λ (σ : L ≃ₐ[K] L), 
    σ.to_alg_hom.comp F.val) f) = pure x.
    \end{lstlisting}
    \vspace{-1em}
    This means that there exists \emph{some} $K$-algebra homomorphism $F \to L$ that generates the ultrafilter $p_*\mathcal{F}$ on $\Hom_K(F,L)$. However, the statement is nonconstructive\footnote{Meaning that there is no explicit formula for this homomorphism.}, so we have to invoke the axiom of choice to take a specific such algebra homomorphism, which is our ${\varphi}_{F,x}$. 
\end{remark}

Subsequently, we glue the local $K$-algebra homomorphisms ${\varphi}_{F,x}$ together to obtain the function $\sigma:L \to L$, defined by:
\begin{lstlisting}
protected noncomputable def 
    ultrafilter.glued_generators_of_pushforwards_function
    (h_int : algebra.is_integral K L) (f : ultrafilter (L →ₐ[K] L)) 
    (x : L) : 
L :=
\end{lstlisting}
Now that we have defined $\sigma$ as a function, we need to define it as a $K$-algebra homomorphism by 
\begin{lstlisting}
noncomputable def 
    ultrafilter.glued_generators_of_pushforwards_alg_hom
    (f : ultrafilter (L →ₐ[K] L)) (h_int : algebra.is_integral K L) : 
L →ₐ[K] L :=
\end{lstlisting}

Next, we prove a lemma, saying that the algebra homomorphism is bijective:
\begin{lstlisting}
lemma ultrafilter.glued_generators_of_pushforwards_alg_hom_bijection 
    (h_int : algebra.is_integral K L) (f : ultrafilter (L →ₐ[K] L)) :
  function.bijective (ultrafilter.glued_generators_of_pushforwards_alg_hom f h_int) :=
\end{lstlisting}
As we saw in Remark~\ref{rmk-function-vs-alg-hom}, {\tt mathlib} considers a $K$-algebra equivalence \lstinline{E ≃ₐ[K] L} to be different from a bijective algebra homomorphism. It consists of two \emph{different} functions $E \to L$ and $L \to E$, together with proofs that they are mutual inverses and that they satisfy the axioms of $K$-algebra homomorphisms. Thankfully, {\tt mathlib} contains a definition, \lstinline{alg_equiv.of_bijective}, which takes a bijective algebra homomorphism and constructs an algebra equivalence whose underlying function equals the underlying function of the algebra homomorphism. We include the statement of \lstinline{alg_equiv.of_bijective} for completeness:
\begin{lstlisting}
noncomputable def alg_equiv.of_bijective (f : A₁ →ₐ[R] A₂) 
    (hf : function.bijective f) : A₁ ≃ₐ[R] A₂ :=
\end{lstlisting}
Now we can finally define $\sigma$ as a term of the Galois group \lstinline{L ≃ₐ[K] L}, as follows:
\begin{lstlisting}
noncomputable def ultrafilter.glued_generators_of_pushforwards_alg_equiv 
    (h_int : algebra.is_integral K L) (f : ultrafilter (L →ₐ[K] L)) :
(L ≃ₐ[K] L) :=
alg_equiv.of_bijective (ultrafilter.glued_generators_of_pushforwards_alg_hom f h_int)
    (ultrafilter.glued_generators_of_pushforwards_alg_hom_bijection 
    h_int f)
\end{lstlisting}
All that remains is to show that this equivalence is actually a limit of the ultrafilter\footnote{Which is just the esoteric way of saying that $\mathcal{N}(\sigma)\subseteq \mathcal{F}$.}, which we do with the following lemma:
\begin{lstlisting}
lemma ultrafilter_converges_to_glued_equiv 
    (h_int : algebra.is_integral K L) (f : ultrafilter (L ≃ₐ[K] L)) :
(f : filter (L ≃ₐ[K] L)) ≤ 
    nhds (ultrafilter.glued_generators_of_pushforwards_alg_equiv h_int
    (f.map (λ (σ : L ≃ₐ[K] L), σ.to_alg_hom))) :=
\end{lstlisting}

At this point, we are pretty much done; our actual proof of compactness is the lemma: 
\begin{lstlisting}
lemma krull_topology_compact {K L : Type*} [field K] [field L] 
    [algebra K L] (h_int : algebra.is_integral K L) :
is_compact (set.univ : set (L ≃ₐ[K] L)) :=
\end{lstlisting}
This is fairly immediate from \lstinline{is_compact_iff_ultrafilter_le_nhds}, which is {\tt mathlib}'s statement of Theorem~\ref{theorem-compact-iff-ultrafilter}. Finally, we prove profiniteness by 
\begin{lstlisting}
def krull_topology_profinite {K L : Type*} [field K] [field L] 
    [algebra K L]  (h_int : algebra.is_integral K L)
    (minpoly K x) :
Profinite := 
{ to_CompHaus := krull_topology_comphaus h_int,
  is_totally_disconnected := 
  krull_topology_totally_disconnected_space h_int}
\end{lstlisting}

\section{Conlcusion and Acknowledgements}
For any field extension $L/K$, not necessarily algebraic, normal, or separable, we defined a canonical topology on the group \lstinline{L ≃ₐ[K] L}, making it into a topological group. This topology generalises the Krull topology, which is typically only defined for Galois extensions. Moreover, we proved that this topology is profinite whenever the extension is algebraic (but not necessarily normal or separable). 

Immense thanks are due to Kevin Buzzard for his support throughout the project. He has been very generous with his time and has written many articles' worth of exposition to me via Zulip messages. More generally, the {\tt mathlib} community has answered any and all questions posed in the Leanprover Zulip server. For anybody starting out in Lean, my top piece of advice is to make use of this community to the fullest. As long as you are demonstrating effort, no question is too basic! 

I am also grateful to Patrick Massot for helping me understand some technical details of filter bases and Thomas Browning for proving the \lstinline{finite_dimensional_sup} lemma, as well as everyone who has commented on my Pull Requests or replied to my Zulip questions.

\printbibliography

@book{conrad, 
author = {Keith Conrad},
title = {Infinite Galois Theory},
year = {2020},
url = {https://ctnt-summer.math.uconn.edu/wp-content/uploads/sites/1632/2020/06/CTNT-InfGaloisTheory.pdf#page=7&zoom=100,144,350}}

@book {maria,
  author          = {Mar\'ia In\'es de Frutos-Fern\'andez},
  title           = {Formalizing the Ring of Ad\`eles of a Global Field},
  publisher       = {Arxiv},
  url = "https://arxiv.org/pdf/2203.16344.pdf",
  year            = {2022}
}

@book {MR2180311,
    AUTHOR = {Howie, John M.},
     TITLE = {Fields and {G}alois theory},
    SERIES = {Springer Undergraduate Mathematics Series},
 PUBLISHER = {Springer-Verlag London, Ltd., London},
      YEAR = {2006},
     PAGES = {x+225},
      ISBN = {978-1-85233-986-9; 1-85233-986-1},
   MRCLASS = {12-01 (12F10)},
  MRNUMBER = {2180311},
MRREVIEWER = {Chandan Singh Dalawat},
}

@book {MR1726779,
    AUTHOR = {Bourbaki, Nicolas},
     TITLE = {General topology. {C}hapters 1--4},
    SERIES = {Elements of Mathematics (Berlin)},
      NOTE = {Translated from the French,
              Reprint of the 1989 English translation},
 PUBLISHER = {Springer-Verlag, Berlin},
      YEAR = {1998},
     PAGES = {vii+437},
      ISBN = {3-540-64241-2},
   MRCLASS = {54-02 (00A05 54-01)},
  MRNUMBER = {1726779},
}

@book {MR1932358,
    AUTHOR = {Dudley, R. M.},
     TITLE = {Real analysis and probability},
    SERIES = {Cambridge Studies in Advanced Mathematics},
    VOLUME = {74},
      NOTE = {Revised reprint of the 1989 original},
 PUBLISHER = {Cambridge University Press, Cambridge},
      YEAR = {2002},
     PAGES = {x+555},
      ISBN = {0-521-00754-2},
   MRCLASS = {60-01 (00A05 28-01 46-01 54-01)},
  MRNUMBER = {1932358},
       DOI = {10.1017/CBO9780511755347},
       URL = {https://doi.org/10.1017/CBO9780511755347},
}

@inbook{berhuy_2010, place={Cambridge}, series={London Mathematical Society Lecture Note Series}, title={Infinite Galois theory}, DOI={10.1017/CBO9781139107051.003}, booktitle={An Introduction to Galois Cohomology and its Applications}, publisher={Cambridge University Press}, author={Berhuy, Grégory}, year={2010}, pages={13–25}, collection={London Mathematical Society Lecture Note Series}}

@book {shatz,
    AUTHOR = {Shatz, Stephen S.},
     TITLE = {Profinite groups, arithmetic, and geometry},
    SERIES = {Annals of Mathematics Studies, No. 67},
 PUBLISHER = {Princeton University Press, Princeton, N.J.; University of
              Tokyo Press, Tokyo},
      YEAR = {1972},
     PAGES = {x+252},
   MRCLASS = {12B20 (12A65 12B25 12G10 14L20)},
  MRNUMBER = {0347778},
MRREVIEWER = {F. Oort},
}
\end{document}